\documentclass[conference]{IEEEtran}
\usepackage{graphicx,indentfirst}
\usepackage{booktabs}
\usepackage{indentfirst}
\usepackage{amsmath}
\usepackage{cite}
\usepackage{multirow}
\usepackage{slashbox}
\usepackage{subfigure}
 \usepackage{mathrsfs}
\usepackage{amsmath}
\usepackage[english]{babel}
\usepackage{amsthm}
\usepackage{multirow}
\usepackage{algorithm, algorithmic}
\usepackage{paralist}
\usepackage{amsthm}
\usepackage{amssymb}
\usepackage{amsfonts}
\usepackage{color}
\newtheorem{proposition}{Proposition}

\newtheorem{assumption}{Assumption}

\IEEEoverridecommandlockouts

\begin{document}
\bibliographystyle{IEEE2}
\title{Competition and Cooperation Analysis for Data Sponsored Market: A Network Effects Model}
\author{Zehui Xiong$^1$, Shaohan Feng$^1$, Dusit Niyato$^1$, Ping Wang$^1$ and Yang Zhang$^2$\\
$^1$School of Computer Science and Engineering, Nanyang Technological University (NTU), Singapore\\
$^2$School of Computer Science and Technology, Wuhan University of Technology, China}
\maketitle
\begin{abstract}
The data sponsored scheme allows the content provider to cover parts of the cellular data costs for mobile users. Thus the content service becomes appealing to more users and potentially generates more profit gain to the content provider. In this paper, we consider a sponsored data market with a monopoly network service provider, a single content provider, and multiple users. In particular, we model the interactions of three entities as a two-stage Stackelberg game, where the service provider and content provider act as the leaders determining the pricing and sponsoring strategies, respectively, in the first stage, and the users act as the followers deciding on their data demand in the second stage. We investigate the mutual interaction of the service provider and content provider in two cases: (i) competitive case, where the content provider and service provider optimize their strategies separately and competitively, each aiming at maximizing the profit and revenue, respectively; and (ii) cooperative case, where the two providers jointly optimize their strategies, with the purpose of maximizing their aggregate profits. We analyze the sub-game perfect equilibrium in both cases. Via extensive simulations, we demonstrate that the network effects significantly improve the payoff of three entities in this market, i.e., utilities of users, the profit of content provider and the revenue of service provider. In addition, it is revealed that the cooperation between the two providers is the best choice for all three entities.
\end{abstract}
\begin{IEEEkeywords}
Data sponsoring, competition and cooperation, network effects, socially-aware service.
\end{IEEEkeywords}

\section{Introduction}\label{Sec:Introduction}
Currently, the demand of cellular data usage continues to rise sharply, and the high data cost becomes one of the critical concerns for Mobile Users (MUs) while consuming the cellular data. Therefore, one of the important challenges for Content Provider (CP) is how to attract more MUs to access its contents and thus achieve a higher revenue gain. In 2014, AT\&T launched a data sponsored scheme~\cite{AT&T}, where the CP (e.g., Youbube, Twitter) can sponsor their MUs' cellular data cost and thereby the MUs access the CP's contents through network Service Provider (SP) with lower charge. Clearly, the data sponsoring potentially creates a triple-win outcome for MUs, CP and SP. Specifically, MUs benefit from consuming cellular data with lower price, which increases the data demand for accessing contents, and in turn the higher demand of MUs contributes to the revenue gain of CP and SP.

With the remarkable interests from academia and industry, the data sponsoring has attracted many researchers to investigate and innovate better schemes.  For example, the authors in~\cite{andrews2014calculating} addressed the issue faced by both CP and SP when the amount of content traffic is uncertain, and derived the pricing by examining traffic data of CP. The sponsoring competition among multiple CPs was studied in~\cite{ma2014subsidization}, and it was also demonstrated that the competition improves the payoff for both SP and CPs. The interaction among the monopoly SP, a single CP and MUs was modeled as a Stackelberg game in~\cite{joe2015sponsoring}, where the MUs are assumed to be homogeneous. Then, the authors in~\cite{zhang2015sponsored} studied the similar problem to that in~\cite{zhang2014sponsoring}, where non-sponsored and sponsored CP coexist. In~\cite{zhang2016tds}, the authors explored the interplay between SP and CPs, and presented a pricing mechanism for sponsored data that is truthful in CP's valuation as well as its underlying traffic. The authors in~\cite{wwb2017sponsor} studied the service-selection process among the MUs as an evolutionary population game and demonstrated that how sponsoring helps to improve the SP's profit and the MU's experience.

However, all of the above works studied the data sponsoring without considering the complex interactions among MUs. The data usage demanded by MUs belongs to the information goods~\cite{brake2016}, and network effect\footnote{Network effect implies that a product or service is more valuable to users as the number of users increases.} is an important phenomenon of information economies~\cite{gong2015network}. The underlying network effects amid the data sponsored market influence the user behaviors, which further complicate the interplay between CP and SP. In this paper, we focus on the competition and cooperation between CP and SP, while taking the network effects among MUs into consideration. The contributions of this paper are as follows:
\begin{itemize}
 \item We model the interplay among the SP, CP and MUs as two-stage game, and we study the pricing, sponsoring strategies as well as data demand, respectively, in data sponsored market through backward induction.
 \item We exploit the local network effects utilizing the structural properties of the underlying social network, which improve the data demand in a large extent.
 \item We consider and analyze the interaction between the SP and CP under non-cooperative and cooperative scenarios, respectively.
  \item Our evaluation results reveal the fact that the cooperation helps achieve the triple-win outcome for three entities, i.e., the MUs, CP and SP.
\end{itemize}

The rest of the paper is structured as follows. We present the system model and formulate a two-stage Stackelberg game in Section~\ref{Sec:Model}. In Section~\ref{Sec:Solution}, we analyze the data demand, the optimal pricing as well as sponsoring using backward induction. Next we present the simulation results in Section~\ref{Sec:Simulation}. Section~\ref{Sec:Conclusion} concludes the paper.
\section{System model and game formulation}\label{Sec:Model}
We consider the data sponsored market consisting of three players: SP, CP and MUs. Their interactions are modeled as a two-stage game with complete information.
\subsection{MUs' data demand}
Consider a set of MUs ${\cal N} \buildrel \Delta \over = \{ 1, \ldots ,N\}$, each MU $i \in {\cal N}$ determines their individual data demand, denoted by $x_i$, where $x_i > 0$. Then, let $\mathbf{x} \buildrel \Delta \over = ({x_1}, \ldots ,{x_N})$ and ${\mathbf{x}}_{-i}$ represent the data demand profile of all the MUs and all other MUs other than MU~$i$, respectively. Then, the utility of MU $i$ is given by:
\begin{multline}\label{Eq:2}
{u_i}({x_i},{\mathbf{x}}_{-i}, {p^{\mathrm{u}}}, {\bf{\theta}}_i) =  {a_i}{x_i} - {b_i}{x_i}^2 + {x_i}\sum\nolimits_{j \in {\cal N}} {{g_{ij}}{x_j}} \\ -{c}{\Big(\sum\nolimits_{j \in \mathcal{N}} {{x_j}} \Big)^2}- {p^{\mathrm{u}}}(1-\theta_i){x_i}.
\end{multline}
The first term represents the internal effects that MU $i$ obtains from consuming the data, for which we use the linear-quadratic function to characterize the decreasing marginal returns~\cite{candogan2012optimal}. $a_i, b_i> 0$ are personal type factors capturing the MU heterogeneity. Similar to~\cite{candogan2012optimal}, the second term ${x_i}{\sum _{j \in \mathcal{N}}}{g_{ij}}{x_j}$ is the external benefits gained from network effects, where $g_{ij}$ denotes the influence of MU $j$ on MU $i$. In social network, the MUs influence each other by social behaviors via their relationships, especially in socially-aware service market~\cite{gong2015network}. In this paper, we assume $g_{ij}=g_{ji}$, which means the social relations are reciprocal.

Since the SP has a limited network capacity, we further introduce the third term to indicate the congestions, and thus we apply the quadratic sum form ${c}{({\sum _{j \in {\cal N}}}{x_j})^2}$, similar to~\cite{gong2015network}. The last term indicates the costs, consisting of the price charged by SP, $p^{\mathrm{u}}$, and the sponsorship provided by CP, $\theta_i$.

\subsection{CP's sponsoring and SP's pricing}
The sponsorship factor $\theta_i$ ($\theta_i \in [0,1]$) for each MU $i$ is decided by the CP. The CP's profit includes an advertisement utility, and a component depending on its sponsorship. The cost of the CP associated with the sponsoring is denoted as ${p^{\mathrm{u}}} \sum\nolimits_{i \in \mathcal {N}} {{x_i}\theta_i}$. Thus, the profit of the CP is formulated as:
\begin{equation}\label{Eq:3}
\mathscr{P} =\gamma \sum\nolimits_{i \in \mathcal{N}} {(s{x_i} - t{x_i}^2)} - {p^{\mathrm{u}}} \sum\nolimits_{i \in \mathcal {N}} {{x_i}{\theta_i}}.
\end{equation}
We also use the linear-quadratic function with the decreasing marginal return property to transform the MUs' data demand to the monetary revenue of CP. $\gamma$ is an adjustable parameter representing the equivalent monetary worth of MUs' data demand, and $s, t> 0$ are coefficients capturing the concavity extent of the function.

The decision variable of the SP is price $p^{\mathrm{u}}$ and we adopt the uniform pricing in this paper\footnote{Generally, the data traffic service fee is the same for all MUs.}. We suppose that the revenue of the SP comes from the payment from the MUs. Then, the revenue of the SP is expressed as:
\begin{equation}\label{Eq:4}
\Pi  = {p^{\mathrm{u}}} \sum\nolimits_{i \in {\cal N}} {{x_i}}.
\end{equation}
\subsection{Two-stage game}
As illustrated in Fig.~\ref{Fig:Framework}, we model the interaction among the SP, CP and MUs as a two-stage game. In Stage~I, the SP and the CP act simultaneously as the leaders of the two-stage game. The SP decides the price $p^{\mathrm{u}}$ to maximize its revenue in~(\ref{Eq:4}), and the CP determines the sponsoring factor ${\bf{\Theta}} = [\theta_1, \theta_2, \ldots, \theta_N]^\top$ to maximize its profit in~(\ref{Eq:3}). 
In Stage~II, the MUs determine the data demand to maximize its individual utility in~(\ref{Eq:2}), acting as the followers of this game. 

\begin{figure}[t]
\centering
\includegraphics[width=.35\textwidth]{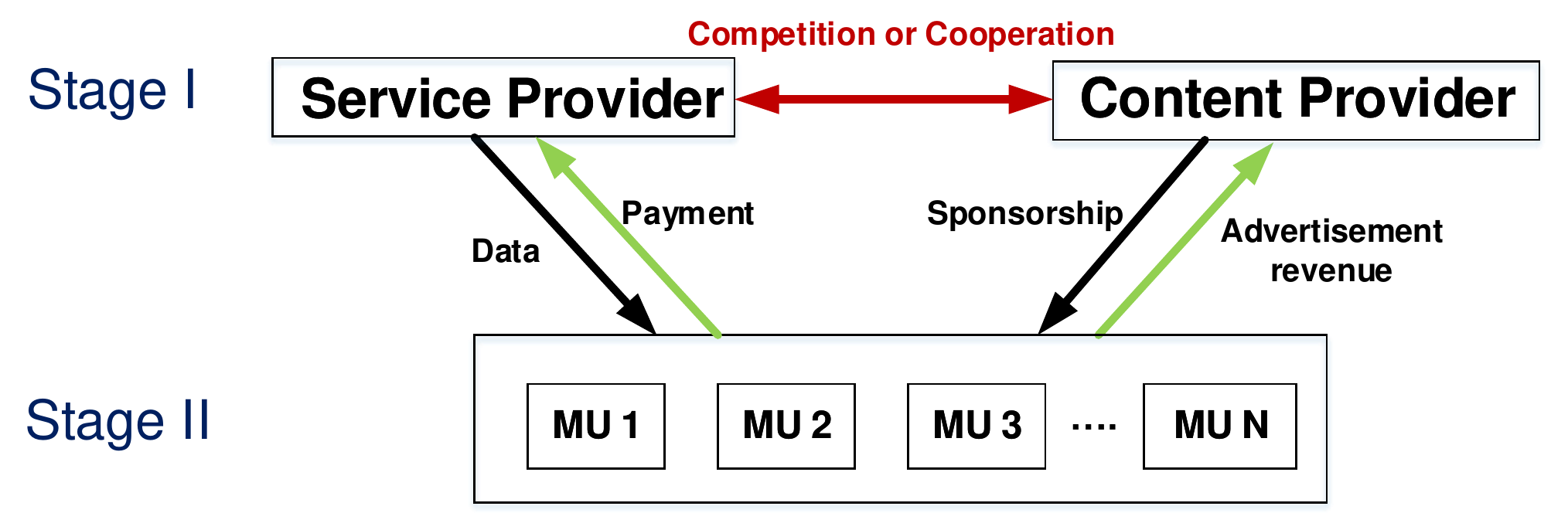}
\caption{Two-stage Stackelberg game model of the interactions among SP, CP, and MUs in the data sponsored market.}\label{Fig:Framework}
\end{figure}

\section{Game equilibrium analysis}\label{Sec:Solution}
\subsection{Stage II: MU data demand equilibrium}
In the sub-game $\mathcal{G}^u= \{\mathcal{N},\{u_i\}_{i \in \mathcal{N}}, [0, +\infty)^N\}$, the best response function of MU $i$ can be obtained in the following proposition according to the first order derivative condition.
\begin{proposition}
Given price $p^{\mathrm{u}}$ and the sponsorship ${\bf{\Theta}}$, and the data demand profile without MU $i$, i.e. ${\bf{x}}_{-i}$, the best response of MU $i$ is obtained as:
\begin{equation}\label{Eq:5}
\mathscr{F}({\bf{x}}_{-i}) = \max\left\{ \frac{{{a_i} - p^{\mathrm{u}}(1 - \theta_i )}}{{2{b_i} + 2c}} + \sum\nolimits_{j \in \mathcal{N}} {{x_j}\frac{{{g_{ij}} - 2c}}{{2{b_i} + 2c}}} \right\}.
\end{equation}

\end{proposition}
Similar to~\cite{candogan2012optimal}, to ensure that each MU has no incentive to unboundedly increase its data demand, we make a general assumption as follows.
\begin{assumption}
$\sum\nolimits_{j \in \mathcal{N}} {\frac{{{g_{ij}} - 2c}}{{2{b_i} + 2c}}} < 1, \forall i$.
\end{assumption}
\begin{proposition}
Under Assumption 1, the existence and the uniqueness of Nash Equilibrium of sub-game $\mathcal{G}^u$ can be established.
\end{proposition}
\begin{proof}
Please refer to~\cite{xiong2017sponsor} for more details.
\end{proof}

For ease of presentation, we define the notations as follows.\begin{footnotesize}
${\bf{G}} = \left[ {\begin{array}{*{20}{c}}
0&{{g_{12}} - 2c}& \cdots &{{g_{1N}} - 2c}\\
{{g_{21}} - 2c}&0& \cdots &{{g_{2N}} - 2c}\\
 \vdots & \vdots & \ddots & \vdots \\
{{g_{N1}} - 2c}&{{g_{N2}} - 2c}& \cdots &0
\end{array}} \right]$,
${\bf{\Lambda}} = \left[ {\begin{array}{*{20}{c}}
{{b_1} + c}&0& \cdots &0\\
0&{{b_2} + c}& \cdots &0\\
 \vdots & \vdots & \ddots & \vdots \\
0&0& \cdots &{{b_N} + c}
\end{array}} \right]$
 \end{footnotesize}, ${\bf{a}} = [a_1, a_2, \ldots, a_N]^\top$, ${\bf{1}} = [1, \ldots, 1]^\top$, ${\bf{I}}$ is the $n \times n$ identity matrix and ${\bf{K}} = (2{\bf{\Lambda }}-{\bf{G}})^{-1}$. Similar to~\cite{xiong2017sponsor}, we consider the ideal situation, as described in the assumption as follows.
\begin{assumption}
All MUs have the positive data demand\footnote{In the typical market, the monopolist seller wants to charge individuals low enough (lower than a threshold), so that all consumers would like to purchase a positive amount of goods~\cite{candogan2012optimal}. Specifically, if MUs are charged appropriately, none of MUs chooses zero data demand.} at the Stackelberg equilibrium, i.e., $x_i>0, \forall i$.
\end{assumption}
To ease the description, we can rewrite~(\ref{Eq:5}) in a matrix form and have the following proposition directly.
\begin{proposition}
The matrix format of best response of all the MUs with respect to the data demand is
\begin{equation}\label{Eq:6}
{\bf{x^*}}({\bf{\Theta}}, p^{\mathrm{u}}) = -{\bf{K}}\left[p^{\mathrm{u}}({\bf{1}} - {\bf{\Theta }}) - {\bf{a}}\right],
\end{equation}where $\bf K$ is positive definite matrix.
\end{proposition}
\begin{figure*}[ht]\scriptsize
\begin{equation*}
\frac{{\partial {\mathscr P}({\bf{\Theta }})}}{{\partial {\bf{\Theta }}}} = \gamma s{p^{\mathrm{u}}}{\bf{K}}{\bf{1}} - 2t\gamma {p^{\mathrm{u}}}{{\bf{K}}^2}[{\bf{a}} - {p^{\mathrm{u}}}({\bf{1}}- {\bf{\Theta }} )] - {p^{\mathrm{u}}}{\bf{K}}{\bf{a}} + {({p^{\mathrm{u}}})^2}{\bf{K}}({\bf{1}} - {\bf{\Theta}}) - {({p^{\mathrm{u}}})^2}{\bf{K}}{\bf{\Theta}} = 0,
\end{equation*}
\begin{equation}\label{Eq1}
\Rightarrow (2t\gamma {p^{\mathrm{u}}}{\bf{K}} + 2{p^{\mathrm{u}}}{\bf{I}})\theta  = \gamma s{\bf{1}} - 2t\gamma {\bf{K}}({\bf{a}} - {p^{\mathrm{u}}}{\bf{1}}) - {\bf{a}} + {p^{\mathrm{u}}}{\bf{1}},
\end{equation}
\begin{equation}\label{Eq:7}
\Rightarrow {\bf{\Theta}}^*(p^{\mathrm{u}}, {\bf{x^*}}) = \frac{1}{{2{p^{\mathrm{u}}}}}{\left(t\gamma {\bf{K}} + {\bf{I}}\right)^{- 1}}\left[\gamma s{\bf{1}} + \left(2t\gamma {\bf{K}} + {\bf{I}}\right)\left({p^{\mathrm{u}}}{\bf{1}} - {\bf{a}}\right)\right].
\end{equation}
\begin{equation}\label{Eq2}
\frac{{\partial \Pi ({p^{\mathrm{u}}})}}{{\partial {p^{\mathrm{u}}}}} = {\bf{1}}^\top{\bf{K}}\left[{\bf{a}} - {p^{\mathrm{u}}}\left({\bf{1}} - {\bf{\Theta}} \right)\right] + {p^{\mathrm{u}}}{\bf{1}}^\top \left[ - {\bf{K}}\left({\bf{1}} - {\bf{\Theta}}\right)\right] = {\bf{1}}^\top {\bf{K}} {\bf{a}} - 2{p^{\mathrm{u}}}{\bf{1}}^\top {\bf{K}}\left({\bf{1}} - {\bf{\Theta}}\right) = 0,
\end{equation}
\begin{equation}\label{Eq:8}
\Rightarrow \{p^{\mathrm{u}}\}^*({\bf{\Theta}}, {\bf{x^*}}) = {\left[2{\bf{1}}^\top {\bf{K}} ({\bf{1}} - {\bf{\Theta}} )\right]^{ - 1}}{\bf{1}}^\top{\bf{K}} {\bf{a}}.
\end{equation}
\hrulefill
\end{figure*}
\subsection{Stage I: CP's sponsoring and SP's pricing strategies}
\subsubsection{Competition between CP and SP}
In this scenario, two entities, i.e., the SP and the CP compete with each other, and their interaction can be modeled as a non-cooperative static game. In this case, the CP aims to maximize its profit and the SP aims to maximize its revenue simultaneously and selfishly.
\begin{proposition}
The existence and the uniqueness of the Nash Equilibrium of the non-cooperative game between the CP and SP can be guanranteed under the Assumption 3.
\end{proposition}
\begin{assumption}
The total payment from MU $i$ to SP is larger than a threshold, i.e., ${p^{\mathrm{u}}}(1 - {\theta _i}) > \max \left\{\gamma s,\frac{{{a_i}}}{3}, {\left[{(\sqrt {2\gamma t{\bf{K}}}  + {\bf{I}})^{ - 1}}\sqrt {\frac{{\gamma t{\bf{K}}}}{2}} {\bf{a}}\right]_i}\right\}.$
\end{assumption}
\begin{proof}
First, using the first order derivative condition, we derive the best response function of CP, given the strategy of SP, which is expressed by~(\ref{Eq:7}). Similarly, given the strategy of CP, we obtain the best response of SP as shown in~(\ref{Eq:8}). From~(\ref{Eq:6}), we can easily derive that $\frac{{\partial {\bf{x}}}}{{\partial {\bf{\Theta}}}} = {p^{\mathrm{u}}}{\bf{K}}$ and $\frac{{\partial {\bf{x}}}}{{\partial {p^{\mathrm{u}}}}} = - {\bf{K}}({\bf{1}} - {\bf{\Theta }})$. The steps of obtaining the best response function of CP and SP are described in~(\ref{Eq1}) and~(\ref{Eq2}), respectively.

\textbf{Existence:} The CP's strategy space is defined to be within $[0, 1]$, which is a nonempty, convex, and compact subset of the Euclidean space. Then, we take the second partial derivative of CP's objective function $\mathscr P({\bf{\Theta }})$ with respect to its own decision variable $\bf{\Theta }$, i.e.,
\begin{equation}
\frac{{{\partial ^2}{\mathscr P}({\bf{\Theta }})}}{{\partial {{\bf{\Theta }}^2}}} =  - 2t\gamma {({p^{\mathrm{u}}})^2}{{\bf{K}}^2} - 2{({p^{\mathrm{u}}})^2}{\bf{K}}<0.
\end{equation}
Therefore, the objective function of CP, ${\mathscr P}(\bf\Theta)$, is continuous and strictly concave with respect to $\bf \Theta$. Similarly, the objective function of SP, $\Pi(p^{\mathrm{u}})$, is strictly concave with respect to its decision variable $p^{\mathrm{u}}$, since we conclude that the second order derivative of $\Pi$ with respect to $p^{\mathrm{u}}$ is less than 0, i.e.,
\begin{equation}
\frac{{\partial^2 {\Pi}({p^{\mathrm{u}}})}}{{\partial {{({p^{\mathrm{u}}})}^2}}} =  - 2{\bf{1}}^\top {\bf{K}}({\bf{1}} - {\bf{\Theta }})<0.
\end{equation}
Recall from Assumption 2, $p^{\mathrm{u}}$ is smaller than a threshold, and $p^{\mathrm{u}}$ is larger than 0. Accordingly, the price $p^{\mathrm{u}}$ is a nonempty convex and compact subset of the Euclidean space. Thus, the Nash equilibrium exists in this non-cooperative sub-game between CP and SP~\cite{han2012game}.

\textbf{Uniqueness:} The Jacobian matrix of point-to-set mapping with respect to the utility profile of CP and SP, $\nabla {\bf{F}}=\nabla {\bf{F}}({\mathscr P}({\bf{\Theta }}), \Pi ({p^{\mathrm{u}}}))=
\left[ {\begin{array}{*{20}{c}}
{\frac{{{\partial ^2}{\mathscr P}({\bf{\Theta }})}}{{\partial {{\bf{\Theta }}^2}}}}&{\frac{{{\partial ^2}{\mathscr P}({\bf{\Theta }})}}{{\partial {\bf{\Theta }}\partial {p^{\mathrm{u}}}}}}\\
{\left(\frac{{{\partial ^2}\Pi ({p^{\mathrm{u}}})}}{{\partial {\bf{\Theta }}\partial {p^{\mathrm{u}}}}}\right)^\top}&{\frac{{{\partial ^2}\Pi ({p^{\mathrm{u}}})}}{{\partial {{\left({p^{\mathrm{u}}}\right)}^2}}}}
\end{array}} \right].$ Then, $\nabla {\bf{F}} + \nabla {\bf{F}}^\top$ is represented by~(\ref{Eq:J}). We decompose the matrix in~(\ref{Eq:J}) into the form of $\rm -A - B - C$, and the matrices are shown in~(\ref{Eq:Decompose1}). If we can conclude that $\nabla {\bf{F}} + \nabla {\bf{F}}^\top$ is negative definite under Assumption 3, then, $\nabla {\bf{F}}$ is diagonally strictly concave, and accordingly the Nash equilibrium of this non-cooperative sub-game is unique~\cite{han2012game}. To prove $\rm -A - B - C$ is negative definite, we can prove that matrices $\rm A$, $\rm B$ and $\rm C$ are all positive definite. Based on the matrix congruence theorem, if $\rm Q' = P^\top Q P$ and $\rm P$ is invertible, then $\rm Q$ and $\rm Q'$ have the same numbers of positive, negative, and zero eigenvalues~\cite{gantmakher1998theory}. We use the matrix $ {\rm P_1} = \left[ {\begin{array}{*{20}{c}}
{\bf{I}}&{{{\left( - 4\gamma t{{({p^{\mathrm{u}}})}^2}{{\bf{K}}^2}\right)}^{ - 1}}2\gamma t{{\bf{K}}^2}\left[{\bf{a}} - 2{p^{\mathrm{u}}}\left({\bf{1}} - {\bf{\Theta }}\right)\right]}\\
0&1
\end{array}} \right]$ and we obtain the congruence of matrix $\rm A$, $\rm A'$ in~(\ref{Eq:A}). Now we only need to prove that $\rm A'$ is positive definite, i.e., $2{{\bf{1}}^ \top }{\bf{K}}({\bf{1}} - {\bf{\Theta }}) - \gamma t{{\bf{K}}^2}{{[{\bf{a}} - 2{p^{\mathrm{u}}}({\bf{1}} - {\bf{\Theta }})]}^ \top }\frac{{{{\bf{K}}^2}}}{{{{({p^{\mathrm{u}}})}^2}}}[{\bf{a}} - 2{p^{\mathrm{u}}}({\bf{1}} - {\bf{\Theta }})] > 0$. Then, we have $2{{\bf{\Theta }}^ \top }{\bf{K}}({\bf{1}} - {\bf{\Theta }}) + 2{\left( {{\bf{1}} - {\bf{\Theta }}} \right)^ \top }{\bf{K}}({\bf{1}} - {\bf{\Theta }}) - {[\frac{{\sqrt {t\gamma } }}{{{p^{\mathrm{u}}
}}}{\bf{a}} - 2\sqrt {t\gamma } ({\bf{1}} - {\bf{\Theta }})]^ \top }{{\bf{K}}^2}[\frac{{\sqrt {t\gamma } }}{{{p^{\mathrm{u}}}}}{\bf{a}} - 2\sqrt {t\gamma } ({\bf{1}} - {\bf{\Theta }})] > 0$. Accordingly, ${\bf{1}} - {\bf{\Theta }} > \sqrt {\bf{K}} [\frac{{\sqrt {t\gamma } }}{{2{p^{\mathrm{u}}}}}{\bf{a}} - 2\sqrt {t\gamma } ({\bf{1}} - {\bf{\Theta }})]$. At last, we derive ${p^{\mathrm{u}}}({\bf{1}} - {\bf{\Theta }}) > \left[ {{{\left(\sqrt {2\gamma t{\bf{K}}}  + {\bf{I}}\right)}^{ - 1}}\sqrt {\frac{{\gamma t{\bf{K}}}}{2}} {\bf{a}}} \right]$. Under the same procedure, we use the matrix ${{\rm{P}}_{\rm{2}}} = \left[ {\begin{array}{*{20}{c}}
{\bf{I}}&{{{\left(3{{\left({p^{\mathrm{u}}}\right)}^2}{\bf{K}}\right)}^{ - 1}}\gamma s{\bf{K1}}}\\
0&1 \end{array}} \right]$ to obtain $\rm B'$ in~(\ref{Eq:B}), then we need to prove that ${{\bf{1}}^ \top }{\bf{K}}({\bf{1}} - {\bf{\Theta }}) - {\gamma ^2}{s^2}{{\bf{1}}^\top }\frac{{\bf{K}}}{{3{{({p^{\mathrm{u}}})}^2}}}{\bf{1}} = {{\bf{\Theta }}^\top }{\bf{K}}({\bf{1}} - {\bf{\Theta }}) + {\left( {{\bf{1}} - {\bf{\Theta }}} \right)^ \top }{\bf{K}}({\bf{1}} - {\bf{\Theta }}) - {\left(\frac{{\gamma s}}{{\sqrt 3 {p^{\mathrm{u}}}}}{\bf{1}}\right)^ \top }{{\bf{K}}^2}\frac{{\gamma s}}{{\sqrt 3 {p^{\mathrm{u}}}}}{\bf{1}} > 0$. Accordingly, we have ${p^{\mathrm{u}}
}({\bf{1}} - {\bf{\Theta }}) > \frac{{\gamma s}}{{\sqrt 3 }}{\bf{1}}$. Similarly, using the matrix ${{\rm{P}}_{\rm{3}}} = \left[ {\begin{array}{*{20}{c}}
{\bf{I}}&{ - {{\left({{({p^{\mathrm{u}}})}^2}{\bf{K}}\right)}^{ - 1}}\left[ {{\bf{a}} - 4{p^{\mathrm{u}}}\left( {{\bf{1}} - {\bf{\Theta }}} \right)} \right]}\\
0&1 \end{array}} \right]$, we have $\rm C'$ in~(\ref{Eq:C}). To prove $\rm C'$ is positive definite, we need to prove ${{\bf{1}}^ \top }{\bf{K}}({\bf{1}} - {\bf{\Theta }}) - {{[{\bf{a}} - 4{p^{\mathrm{u}}}({\bf{1}} - {\bf{\Theta }})]}^ \top }\frac{{\bf{K}}}{{{{({p^{\mathrm{u}}})}^2}}}[{\bf{a}} - 4{p^{\mathrm{u}}
}({\bf{1}} - {\bf{\Theta }})] > 0$. With simple transformations, we have ${p^{\mathrm{u}}}({\bf{1}} - {\bf{\Theta }}) > \frac{{\bf{a}}}{5}$ in the final step. Accordingly, we have proved that the positive definiteness of $\rm A$, $\rm B$ and $\rm C$ are easily guaranteed, if Assumption 3 holds. Thus, the proof is completed.
\begin{figure*}[ht]\scriptsize
\begin{equation}\label{Eq:J}
\left[ {\begin{array}{*{20}{c}}
{ - 4t\gamma {{({p^{\mathrm{u}}})}^2}{{\bf{K}}^2} - 4{{({p^{\mathrm{u}}})}^2}{\bf{K}}}&{\gamma s{\bf{K}}{\bf{1}} - 2\gamma t{{\bf{K}}^2}\left[{\bf{a}} - 2{p^{\mathrm{u}}}\left({\bf{1}} - {\bf{\Theta }}\right)\right] - {\bf{K}}\left[{\bf{a}} - 4{p^{\mathrm{u}}}({\bf{1}} - {\bf{\Theta }})\right]}\\
\left( \gamma s{\bf{K1}} - 2\gamma t{{\bf{K}}^2}\left[{\bf{a}} - 2{p^{\mathrm{u}}}({\bf{1}} - {\bf{\Theta }})\right] - {\bf{K}}\left[{\bf{a}} - 4{p^{\mathrm{u}}}({\bf{1}} - {\bf{\Theta }})\right] \right)^\top &{ - 4{{\bf{1}}^ \top }{\bf{K}}({\bf{1}} - {\bf{\Theta }})}
\end{array}} \right]
\end{equation}

\begin{equation}\label{Eq:Decompose1}
\begin{footnotesize}
\underbrace {\left[ {\begin{array}{*{20}{c}}
{4t\gamma {{({p^{\mathrm{u}}})}^2}{{\bf{K}}^2}}&{2\gamma t{{\bf{K}}^2}[{\bf{a}} - 2{p^{\mathrm{u}}}({\bf{1}} - {\bf{\Theta }})]}\\
{2\gamma t{{[{\bf{a}} - 2{p^{\mathrm{u}}}({\bf{1}} - {\bf{\Theta }})]}^ \top }{{\bf{K}}^2}}&{2{{\bf{1}}^ \top }{\bf{K}}({\bf{1}} - {\bf{\Theta }})}
\end{array}} \right]}_{\rm A}, \underbrace {\left[ {\begin{array}{*{20}{c}}
{3{{({p^{\mathrm{u}}})}^2}{\bf{K}}}&{ - \gamma s{\bf{K}}{\bf{1}}}\\
{ - \gamma s{{\bf{1}}^ \top }{\bf{K}}}&{{{\bf{1}}^ \top }{\bf{K}}({\bf{1}} - {\bf{\Theta }})}
\end{array}} \right]}_{\rm B}, \underbrace {\left[ {\begin{array}{*{20}{c}}
{{{({p^{\mathrm{u}}})}^2}{\bf{K}}}&{{\bf{K}}[{\bf{a}} - 4{p^{\mathrm{u}}}({\bf{1}} - {\bf{\Theta }})]}\\
{{{[{\bf{a}} - 4{p^{\mathrm{u}}}({\bf{1}} - {\bf{\Theta }})]}^ \top }{\bf{K}}}&{{{\bf{1}}^ \top }{\bf{K}}({\bf{1}} - {\bf{\Theta }})}
\end{array}} \right]}_{\rm C}
\end{footnotesize}
\end{equation}

\begin{equation}\label{Eq:A}
\rm A' = P_1^\top A P_1 = \left[ {\begin{array}{*{20}{c}}
{4t\gamma {{({p^{\mathrm{u}}})}^2}{{\bf{K}}^2}}&{\bf{0}}\\
{\bf{0}}&{2{{\bf{1}}^ \top }{\bf{K}}({\bf{1}} - {\bf{\Theta }}) - \gamma t{{\bf{K}}^2}{{[{\bf{a}} - 2{p^{\mathrm{u}}}({\bf{1}} - {\bf{\Theta }})]}^ \top }\frac{{{{\bf{K}}^2}}}{{{{({p^{\mathrm{u}}})}^2}}}[{\bf{a}} - 2{p^{\mathrm{u}}}({\bf{1}} - {\bf{\Theta }})]}
\end{array}} \right]
\end{equation}

\begin{equation}\label{Eq:B}
\rm B' = P_2^\top B P_2 = \left[ {\begin{array}{*{20}{c}}
{3{{({p^{\mathrm{u}}})}^2}{\bf{K}}}&{\bf{0}}\\
{\bf{0}}&{{{\bf{1}}^ \top }{\bf{K}}({\bf{1}} - {\bf{\Theta }}) - {\gamma ^2}{s^2}{{\bf{1}}^ \top }\frac{{\bf{K}}}{{3{{({p^{\mathrm{u}}})}^2}}}{\bf{1}}}
\end{array}} \right]
\end{equation}

\begin{equation}\label{Eq:C}
\rm C' = P_3^\top C P_3 = \left[ {\begin{array}{*{20}{c}}
{{{({p^{\mathrm{u}}})}^2}{\bf{K}}}&{\bf{0}}\\
{\bf{0}}&{{{\bf{1}}^ \top }{\bf{K}}({\bf{1}} - {\bf{\Theta }}) - {{[{\bf{a}} - 4{p^{\mathrm{u}}}({\bf{1}} - {\bf{\Theta }})]}^ \top }\frac{{\bf{K}}}{{{{({p^{\mathrm{u}}})}^2}}}[{\bf{a}} - 4{p^{\mathrm{u}}}({\bf{1}} - {\bf{\Theta }})]}
\end{array}} \right]
\end{equation}

\hrulefill
\end{figure*}
\end{proof}

Then, we use the best-response dynamics for calculating the Nash equilibrium of the two-player non-cooperative game in this stage. So far, we have proved that each Nash equilibrium of sub-game in this Stackelberg game is unique under Assumptions 1 and 3. Then, we can conclude that this Stackelberg game is a weakly acyclic game, and thus the convergence of the game is guaranteed~\cite{marden2007regret}. Therefore, the best-response dynamics algorithm can achieve the Stackelberg equilibrium.

\subsubsection{Cooperation between CP and SP}\label{solution}
In the non-cooperative game, the interaction among selfish players may lead to an inefficient Nash Equilibrium (NE). In order to address the well-known inefficiency of NE of the non-cooperative game, we consider another practical cooperative setting between the CP and the SP. In this case, the interaction between the two providers is modeled as optimization problem. Thus, the objectives of the CP and SP are to maximize their aggregate payoff.

In particular, we consider the SP and CP as a single entity, referred to as a coalition. Then, in Stage~I, the CP-SP coalition determines the sponsoring and the pricing strategies jointly, with the purpose of maximizing their aggregate payoff, i.e., ${\mathscr R} = {\mathscr P} + \Pi$. Therefore, the CP-SP coalition's payoff maximization problem is formulated as follows:
\begin{equation}
\begin{aligned}
& \underset{\theta_i, p^{\mathrm{u}}}{\text{maximize}}
& & {\mathscr R} = \gamma \sum\limits_{i \in \mathcal{N}} {(s{x_i} - t{x_i}^2)} + {p^{\mathrm{u}}} \sum\limits_{i \in \mathcal {N}} {{x_i}{(1- \theta_i)}} \\
& \text{subject to}
& & {\bf{x}} = -{\bf{K}}\left[p^{\mathrm{u}}({\bf{1}} - {\bf{\Theta }}) - {\bf{a}}\right].\\
\end{aligned}\label{Eq:9}
\end{equation}

We can rewrite the objective function of~(\ref{Eq:9}) in a matrix form, and eliminate ${\bf{x}}$ from objective function with KKT condition. Then, we apply the second-order partial derivative to check its Hessian matrix. Thus, we conclude the following proposition.
\begin{proposition}
Under Assumption 3, the objective function in~(\ref{Eq:9}) is strictly concave with respect to its both decision variables $\bf \Theta$ and $p^{\mathrm{u}}$, and thus there exists unique globally optimal solutions for $\left\{{\bf {\Theta^*}}, \{p^{\mathrm{u}}\}^*\right\}$.
\end{proposition}
\begin{proof}
\begin{figure*}[ht]\scriptsize
\begin{equation}\label{Eq:H}
\left[ {\begin{array}{*{20}{c}}
{ - 2t\gamma {{({p^{\mathrm{u}}})}^2}{{\bf{K}}^2} - 2{{({p^{\mathrm{u}}})}^2}{{\bf{K}}^2}}&{\gamma s{\bf{K1}} - 2\gamma t{{\bf{K}}^2}\left[ {{\bf{a}} - 2{p^{\mathrm{u}}}\left( {{\bf{1}} - {\bf{\Theta }}} \right)} \right] - {\bf{K}}\left[ {{\bf{a}} - 4{p^{\mathrm{u}}}\left( {{\bf{1}} - {\bf{\Theta }}} \right)} \right]}\\
{{{\left( {\gamma s{\bf{K1}} - 2\gamma t{{\bf{K}}^2}\left[ {{\bf{a}} - 2{p^{\mathrm{u}}}\left( {{\bf{1}} - {\bf{\Theta }}} \right)} \right] - {\bf{K}}\left[ {{\bf{a}} - 4{p^{\mathrm{u}}}\left( {{\bf{1}} - {\bf{\Theta }}} \right)} \right]} \right)}^ \top }}&{ - 2t\gamma {{\left( {{\bf{1}} - {\bf{\Theta }}} \right)}^ \top }{{\bf{K}}^2}\left( {{\bf{1}} - {\bf{\Theta }}} \right) - 2{{\left( {{\bf{1}} - {\bf{\Theta }}} \right)}^ \top }{\bf{K}}\left( {{\bf{1}} - {\bf{\Theta }}} \right)}
\end{array}} \right]
\end{equation}
\begin{equation}\label{Eq:Decompose2}
 \underbrace {\left[ {\begin{array}{*{20}{c}}
{2t\gamma {{({p^{\mathrm{u}}})}^2}{{\bf{K}}^2}}&{2\gamma t{{\bf{K}}^2}\left[ {{\bf{a}} - 2{p^{\mathrm{u}}}\left( {{\bf{1}} - {\bf{\Theta }}} \right)} \right]}\\
{2\gamma t{{[{\bf{a}} - 2{p^{\mathrm{u}}}\left( {{\bf{1}} - {\bf{\Theta }}} \right)]}^ \top }{{\bf{K}}^2}}&{2\gamma t{{\left( {{\bf{1}} - {\bf{\Theta }}} \right)}^ \top }{{\bf{K}}^2}\left( {{\bf{1}} - {\bf{\Theta }}} \right)}
\end{array}} \right]}_{\rm{D}}, \underbrace {\left[ {\begin{array}{*{20}{c}}
{{{({p^{\mathrm{u}}})}^2}{\bf{K}}}&{{\bf{K}}\left[ {{\bf{a}} - 4{p^{\mathrm{u}}}\left( {{\bf{1}} - {\bf{\Theta }}} \right)} \right]}\\
{{{\left[ {{\bf{a}} - 4{p^{\mathrm{u}}}\left( {{\bf{1}} - {\bf{\Theta }}} \right)} \right]}^ \top }{\bf{K}}}&{{{\left( {{\bf{1}} - {\bf{\Theta }}} \right)}^ \top }{\bf{K}}\left( {{\bf{1}} - {\bf{\Theta }}} \right)}
\end{array}} \right]}_{\rm{E}}, \underbrace {\left[ {\begin{array}{*{20}{c}}
{{{({p^{\mathrm{u}}})}^2}{\bf{K}}}&{ - \gamma s{\bf{K1}}}\\
{ - \gamma s{{\bf{1}}^ \top }{\bf{K}}}&{{{\left( {{\bf{1}} - {\bf{\Theta }}} \right)}^ \top }{\bf{K}}\left( {{\bf{1}} - {\bf{\Theta }}} \right)}
\end{array}} \right]}_{\rm{F}}
\end{equation}

\begin{equation}\label{Eq:D}
{\rm{D' = }}{{\rm{P_4}}^ \top }{\rm{DP_4 = }}\left[ {\begin{array}{*{20}{c}}
{{{2t}}\gamma {{({{{p}}^{{u}}})}^{{2}}}{{\bf{K}}^{{2}}}}&{\bf{0}}\\
{\bf{0}}&{{{2t}}\gamma {{({\bf{1}}{{ - }}{\bf{\Theta }})}^ \top }{{\bf{K}}^{{2}}}({\bf{1}}{{ - }}{\bf{\Theta }}){{ - 2}}\gamma {{t}}{{\left[{\bf{a}}{{ - 2}}{{{p}}^{{u}}}\left( {{\bf{1}} - {\bf{\Theta }}} \right)\right]}^ \top }\frac{{{{\bf{K}}^{{2}}}}}{{{{({{{p}}^{{u}}})}^{{2}}}}}\left[{\bf{a}}{{ - 2}}{{{p}}^{{u}}}\left( {{\bf{1}} - {\bf{\Theta }}} \right)\right]}
\end{array}} \right]
\end{equation}

\begin{equation}\label{Eq:E}
{\rm{E' = }}{{\rm{P_5}}^ \top }{\rm{EP_5= }}\left[ {\begin{array}{*{20}{c}}
{{{({{{p}}^{{u}}})}^{{2}}}{\bf{K}}}&{\bf{0}}\\
{\bf{0}}&{{{\left( {{\bf{1}} - {\bf{\Theta }}} \right)}^ \top }{\bf{K}}\left( {{\bf{1}} - {\bf{\Theta }}} \right){{ - }}{{[{\bf{a}}{{ - 4}}{{{p}}^{{u}}}\left({{\bf{1}} - {\bf{\Theta }}} \right)]}^ \top }\frac{{\bf{K}}}{{{{({{{p}}^{{u}}})}^{{2}}}}}\left[{\bf{a}}{{ - 4}}{{{p}}^{{u}}}\left( {{\bf{1}} - {\bf{\Theta }}} \right)\right]}
\end{array}} \right]
\end{equation}

\begin{equation}\label{Eq:F}
{\rm{F' = }}{{\rm{P_6}}^ \top }{\rm{FP_6 = }}\left[{\begin{array}{*{20}{c}}
{{{({{{p}}^{{u}}})}^{{2}}}{\bf{K}}}&{\bf{0}}\\
{\bf{0}}&{{{\left( {{\bf{1}} - {\bf{\Theta }}} \right)}^ \top }{\bf{K}}\left( {{\bf{1}} - {\bf{\Theta }}} \right) - {\gamma ^{{2}}}{{{s}}^{{2}}}{{\bf{1}}^ \top }\frac{{\bf{K}}}{{{{({{{p}}^{{u}}})}^{{2}}}}}{\bf{1}}}
\end{array}} \right]
\end{equation}

\hrulefill
\end{figure*}

The Hessian matrix of objective function ${\mathscr R}({\bf \Theta }, p^{\mathrm{u}})$ can be obtained from $\left[ {\begin{array}{*{20}{c}}
{\frac{{{\partial ^2}\mathscr R}}{{\partial {{\bf{\Theta}} ^2}}}}&{\frac{{{\partial ^2}\mathscr R}}{{\partial {\bf{\Theta}} \partial {p^{\mathrm{u}}}}}}\\
{{{\left( {\frac{{{\partial ^2}\mathscr R}}{{\partial {\bf{\Theta}} \partial {p^{\mathrm{u}}}}}} \right)}^ \top }}&{\frac{{{\partial ^2}\mathscr R}}{{\partial {{\left( {{p^{\mathrm{u}}}} \right)}^2}}}}
\end{array}} \right]$, as shown in~(\ref{Eq:H}). Similar to that in the proof of Proposition 4, we decompose the Hessian matrix into three matrix $\rm -D-E-F$, and the matrices are shown in~(\ref{Eq:Decompose2}). To prove that~(\ref{Eq:H}) is negative definite, we can simply prove that $\rm D$, $\rm E$ and $\rm F$ are all positive definite. Here we also use the matrix congruence theorem, to prove the positivity of its congruence matrix. Specifically, we use $P_4 = \left[ {\begin{array}{*{20}{c}} {\bf{I}}&{ - {{\left(2t\gamma {{({p^{\mathrm{u}}})}^2}{{\bf{K}}^2}\right)}^{ - 1}}2t\gamma {{\bf{K}}^2}\left[ {{\bf{a}} - 2{p^{\mathrm{u}}}\left( {{\bf{1}} - {\bf{\Theta }}} \right)} \right]}\\
{{0}}&1 \end{array}} \right]$ to obtain $\rm D'$ in~(\ref{Eq:D}). To prove that $\rm D'$ is positive definite, we have $ - 2\gamma t{\left[ {\frac{{\bf{a}}}{{{p^{\mathrm{u}}}}} - 2\left( {{\bf{1}} - {\bf{\Theta }}} \right)} \right]^ \top }{{\bf{K}}^2}\left[ {\frac{{\bf{a}}}{{{p^{\mathrm{u}}}}} - 2\left( {{\bf{1}} - {\bf{\Theta }}} \right)} \right] + 2t\gamma {({\bf{1}} - {\bf{\Theta }})^ \top }{{\bf{K}}^2}({\bf{1}} - {\bf{\Theta }}) > 0$. Moreover, we have $\left( {{\bf{1}} - {\bf{\Theta }}} \right) > \frac{{\bf{a}}}{{{p^{\mathrm{u}}}}} - 2\left( {{\bf{1}} - {\bf{\Theta }}} \right)$, which implies ${p^{\mathrm{u}}}\left( {{\bf{1}} - {\bf{\Theta }}} \right) > \frac{{\bf{a}}}{3}$. Then, we use $P_5 = \left[ {\begin{array}{*{20}{c}}
{\bf{I}}&{ - {{\left({{({p^{\mathrm{u}}})}^2}{\bf{K}}\right)}^{ - 1}}{\bf{K}}\left[ {{\bf{a}} - 4{p^{\mathrm{u}}}\left( {{\bf{1}} - {\bf{\Theta }}} \right)} \right]}\\
{{0}}&1 \end{array}} \right]$ and $P_6 = \left[ {\begin{array}{*{20}{c}}
{\bf{I}}&{{{\left({{({p^{\mathrm{u}}})}^2}{\bf{K}}\right)}^{ - 1}}\gamma s{\bf{K1}}}\\
{{0}}&1 \end{array}} \right]$ to obtain $\rm E'$ in~(\ref{Eq:E}) and $\rm F'$ in~(\ref{Eq:F}). Similarly, to prove that $\rm E'$ and $\rm F'$ are positive definite, we have to ensure ${{{\left( {{\bf{1}} - {\bf{\Theta }}} \right)}^ \top }{\bf{K}}\left( {{\bf{1}} - {\bf{\Theta }}} \right) - {\gamma ^2}{s^2}{{\bf{1}}^ \top }\frac{{\bf{K}}}{{{{({p^{\mathrm{u}}})}^2}}}{\bf{1}}}>0$ and $- {{[{\bf{a}} - 4{p^{\mathrm{u}}}\left( {{\bf{1}} - {\bf{\Theta }}} \right)]}^ \top }\frac{{\bf{K}}}{{{{({p^{\mathrm{u}}})}^2}}}\left[ {{\bf{a}} - 4{p^{\mathrm{u}}}\left( {{\bf{1}} - {\bf{\Theta }}} \right)} \right] + {{\left( {{\bf{1}} - {\bf{\Theta }}} \right)}^ \top }{\bf{K}}\left( {{\bf{1}} - {\bf{\Theta }}} \right)  > 0$, which corresponds to ${p^{\mathrm{u}}}\left( {{\bf{1}} - {\bf{\Theta }}} \right) > \gamma s{\bf{1}}$ and ${p^{\mathrm{u}}}\left( {{\bf{1}} - {\bf{\Theta }}} \right) > \frac{{\bf{a}}}{5}$, respectively. We can see the condition where the Hessian matrix of $\mathscr R$ is negative definite holds under Assumption~3. The proof is completed.
\end{proof}
However, it is impossible to derive the closed form solution for $\bf \Theta^*$ and $\{p^{\mathrm{u}}\}^*$, due to their complicated expression. In our performance evaluation, we can apply the low-complexity iterative algorithms based on the gradient assisted binary search algorithm to find the optimal sponsorship $\bf \Theta^*$ and optimal price $\{p^{\mathrm{u}}\}^*$, which are the optimal solutions in Stage~I. In all, the $\left\{{\bf{x^*}}, {\bf{\Theta^*}}, \{p^{\mathrm{u}}\}^* \right\}$ solved in Section~\ref{solution} is the Stackelberg equilibrium for this two-stage game.

\section{Simulation results}\label{Sec:Simulation}
In this section, we perform the simulations to evaluate the performance of the strategy adaptation of the SP, CP, and MUs in data sponsored market under the competitive and cooperative cases. We consider a group of $N$ MUs in a social network and assume that the parameters of MUs $a_i$ and $b_i$ follow the normal distribution $\mathcal{N}(\mu_a, 1)$ and $\mathcal{N}(\mu_b, 1)$. Likewise, the social tie $g_{ij}$ between any two MUs $i$ and $j$ follows a normal distribution $\mathcal{N}(\mu_g, 1)$. The default parameters are set as follows: $c=3$, $\gamma=2$, $s=t=5$, $\mu_a=\mu_b=30$, $\mu_g = 4$ and $N=100$.

We first investigate the impact from varying the number of MUs on the three entities of data sponsored market, i.e., MUs, CP and SP, as shown in Fig.~\ref{Fig:number}. As expected, the total data demand of MUs, the profit of CP and the revenue of SP increase as the number of MUs increases, in either competitive or cooperative case. The reason is that adding more MUs would enhance each MU's interactions with others, and potentially stimulate more data demand of new coming MUs. However, owing to the congestion effects, the marginal increase of the data demand decreases as the number of MUs increases. Meanwhile, the CP provides more sponsorship for the coming MUs. We also observe that the optimal price increases with the increase of number of MUs. This is because as the number of MUs increases, more MUs have higher intrinsic demands, so that increasing the price does not result in significant decrease in total demand. Thus, both the CP and SP extract more surplus from MUs and have higher payoff.

\begin{figure}[t]
\centering
\includegraphics[width=.43\textwidth]{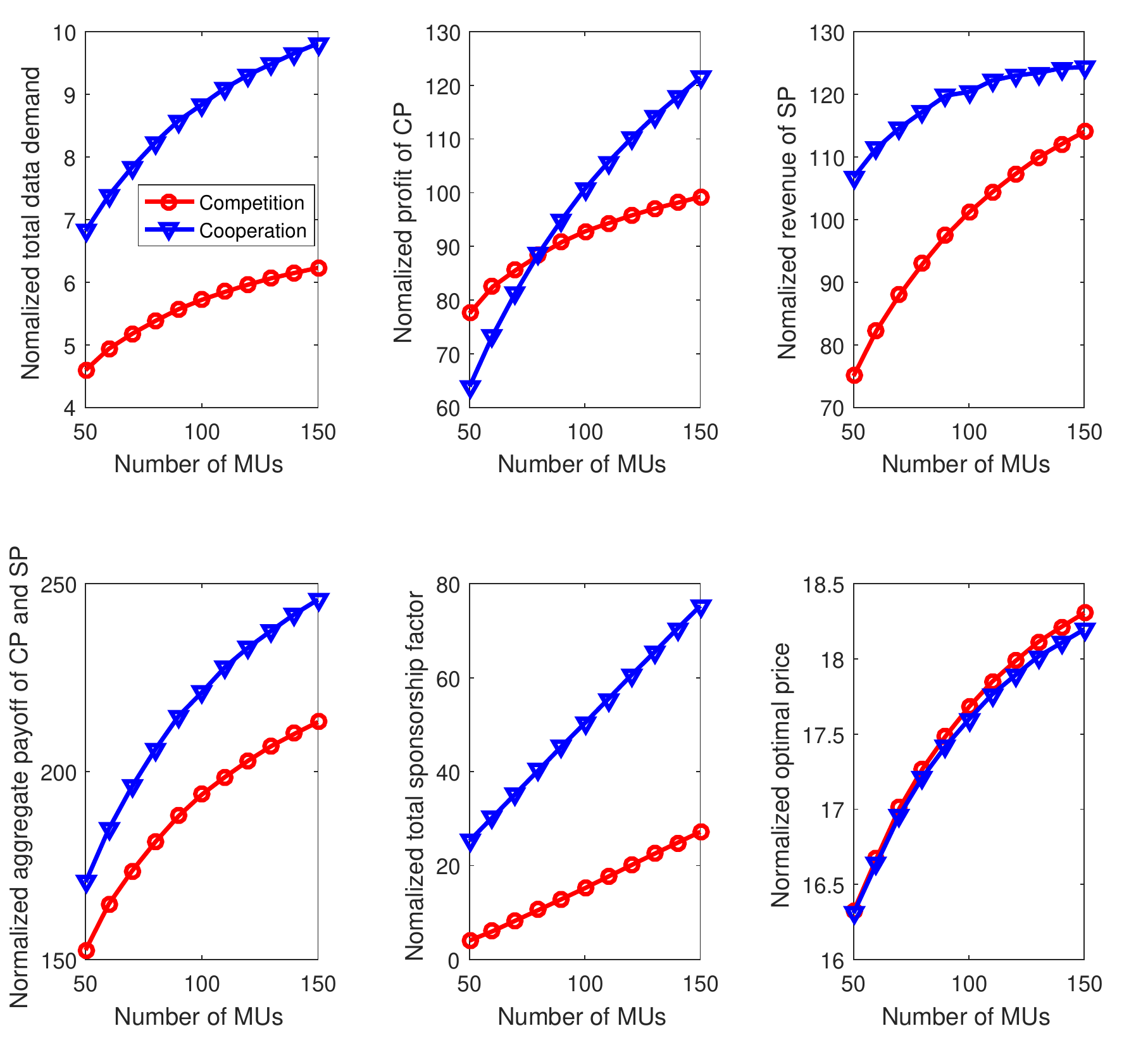}
\caption{The impact of the total number of MUs on three entities of data sponsored market.}\label{Fig:number}
\end{figure}

We then investigate the impact of network effects on these three entities, and the results are shown in Fig.~\ref{Fig:networkeffect}. In both competitive and cooperative cases, the total data demand of MUs, the profit of CP and the revenue of SP increase significantly with stronger network effects. Additionally, as the network effects become stronger, the demand of each MU is promoted due to stronger positive interdependency of each other. When the demand of MUs is high enough, consequently the CP is able to provide lower sponsorship to save money. We find that the optimal price increases as the network effects becomes stronger. The reason is that when the network effect is stronger, the additional benefits obtained due to network effects are greater. To take advantage of the underlying network effects, the SP sets a lower price to encourage more MUs, and each MU will be influenced by neighbours for higher data demand. Thus, the SP achieves higher revenue correspondingly.

\begin{figure}[t]
\centering
\includegraphics[width=.45\textwidth]{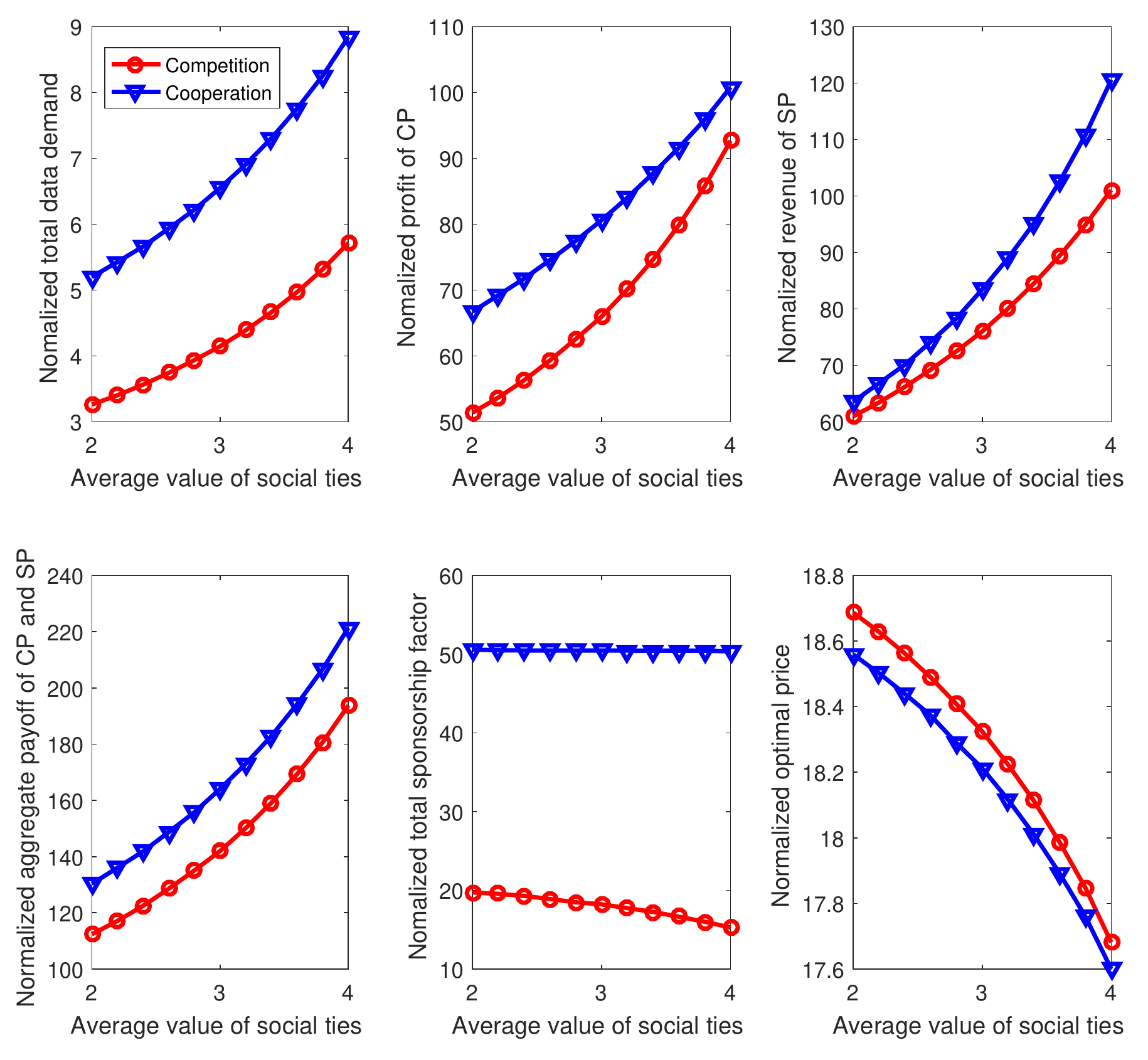}
\caption{The impact of social network effects on three entities of data sponsored market.}\label{Fig:networkeffect}
\end{figure}

\begin{figure}[t]
\centering
\includegraphics[width=.45\textwidth]{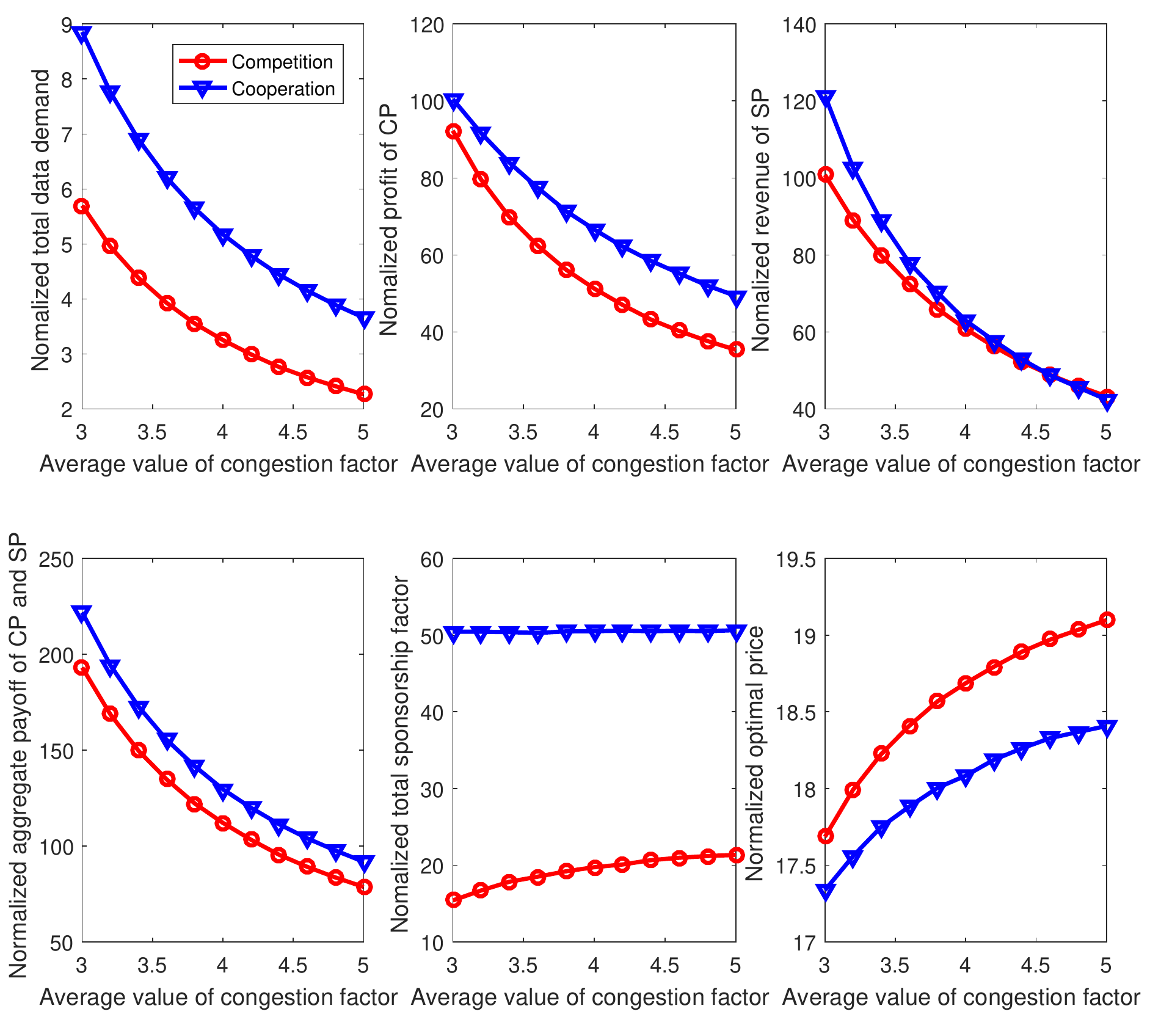}
\caption{The impact of congestion effects on three entities of data sponsored market.}\label{Fig:congestion}
\end{figure}

Next, we evaluate the impact of congestion on three entities, as illustrated in Fig.~\ref{Fig:congestion}. We observe that the total data demand of MUs decreases as the congestion increases. This is because the congestion has a negative impact on the data demand of MUs. Then, the SP does not need to set a lower price to attract more MUs, which may decrease the total data demand further. At the same time, the CP needs to provide more sponsorship with the increase of price to retain the original MUs at least, which may incur more costs. Accordingly, the profit of CP and the revenue of SP decrease as the congestion increases.

At last, we study the impact of competition and cooperation between CP and SP in Figs.~\ref{Fig:number}-\ref{Fig:congestion}. Recall that in cooperative case, the CP-SP coalition aims to maximize their aggregate payoff. In other words, the coalition is able to adopt the discriminatory pricing in this market. In this case, the coalition usually sets a lower price, and provides more and appropriate distribution of sponsorship to better encourage all the MUs. Therefore, the cooperation between the CP and SP helps to achieve higher aggregate payoff, as indicated in Figs.~\ref{Fig:number}-\ref{Fig:congestion}. However, in the competitive case, when SP sets a lower price, the CP provides less sponsorship to maximize its profit because the sponsorship is not necessary when the price is low enough. As a result, the aggregate payoff of CP and SP decreases, compared with the cooperative case. All the results in Figs.~\ref{Fig:number}-\ref{Fig:congestion} clearly show that the cooperation between CP and SP is the best choice for all of three entities in the data sponsored market.

\section{Conclusion}\label{Sec:Conclusion}
We have proposed a two-stage game to model the data sponsored market. In particular, we have focused on the interaction between a single service provider and a single content provider. Further, we have taken the impacts of network effects into account to reconsider the competition and the cooperation in the data sponsored market. Specifically, the network effects in the social domain and congestions in the network domain have been jointly considered for modeling the interactions among mobile users. Under this setting, we have characterized the scenario where the service provider and content provider compete separately, and the scenario where the service provider and content provider cooperate for a common goal. Finally, through extensive simulations, it has been verified that the cooperation is the best choice for three entities, i.e., the service provider, content provider and mobile users. For the future work, we will further consider the unfixed payment from the content provider to the service provider.
\bibliography{bibfile}

\end{document}